%% file: main.tex
\documentclass[letterpaper, 10pt, conference,usenames,dvipsnames]{IEEEtran}
\IEEEoverridecommandlockouts
\usepackage{amssymb, amsmath, latexsym, amsfonts, mathrsfs, amsbsy,mathtools,relsize}
\usepackage[english]{babel}

\newcommand{\TC}[1]{\textcolor{cyan}{{#1}}}

\usepackage{subcaption}

\usepackage[font=small,labelfont=bf]{caption}
\usepackage[dvips]{graphicx}
\usepackage[cp1252, latin1]{inputenc}
\usepackage{epsfig, color}
\usepackage[sort]{cite}
\usepackage{listings}
\usepackage{subfloat}
\usepackage{setspace}
\usepackage{makeidx, glossaries}
\usepackage{url}
\usepackage{changes}
\usepackage{stmaryrd}
\usepackage{enumerate}
\usepackage{epigraph}
\usepackage{multirow}
\usepackage{footnote}
\usepackage{bbm}
\usepackage{pgfplots}
\usepackage[T1]{fontenc}
\usepackage{calligra}
\usepackage{soul}
\usepackage{comment}

\usepackage{algorithm}
\usepackage[noend]{algpseudocode}
\def\BState{\State\hskip-\ALG@thistlm}
\makeatother

\newcommand{\CASE}[1]{\STATE \textbf{case} #1\textbf{:} \begin{ALC@g}}
\newcommand{\ENDCASE}{\end{ALC@g}}

\newcommand{\DEFAULT}{\STATE \textbf{default:} \begin{ALC@g}}
\newcommand{\ENDDEFAULT}{\end{ALC@g}}
\newcommand{\DEFAULTLINE}[1]{\STATE \textbf{default:} }
\usepackage{lipsum}

\usepackage{changes}
\usepackage{lipsum}% <- For dummy text
\definechangesauthor[name={Gabriele}, color=orange]{gab}
\usepackage{physics}
\usepackage{color}
\usepackage{array,booktabs,calc}
\usepackage{algorithm}
\usepackage{algpseudocode}
\algnewcommand{\LineComment}[1]{\State \(\triangleright\) #1}
\makeatletter
\def\BState{\State\hskip-\ALG@thistlm}
\makeatother
\usepackage{float}
\usepackage{graphics} % for pdf, bitmapped graphics files
\usepackage{epsfig} % for postscript graphics files
\usepackage{multirow}
\usepackage{amsmath}
\usepackage{setspace}
\usepackage{amsbsy}
\usepackage{supertabular}
\usepackage{tkz-graph}
\usepackage{etex}
\usepackage{pgf, tikz}
\usepackage{bm}
\usetikzlibrary{shapes,shapes.geometric,arrows,fit,calc,positioning,automata,through,intersections}
\tikzset{diamond state/.style={draw,diamond}}

\usepackage{dsfont}

\usepackage{capt-of}
\usepackage{calligra}
\usepackage{algorithm}
\usepackage{algpseudocode}
\usepackage{setspace}
\usepackage{subfiles}

\usepackage{pgf, tikz}
\usetikzlibrary{shapes,shapes.geometric,arrows,fit,calc,positioning,automata,through,intersections}
\tikzset{diamond state/.style={draw,diamond}}
\usetikzlibrary{positioning,calc}
\usepackage{smartdiagram}
\usepackage[algo2e]{algorithm2e} 
\usepackage{lipsum}

\newtheorem{theorem}{Theorem}
\newtheorem{remark}{Remark}

\newtheorem{assumption}{Assumption}

\algdef{SE}[SWITCH]{Switch}{EndSwitch}[1]{\algorithmicswitch\ #1\ \algorithmicdo}   {\algorithmicend\ \algorithmicswitch}%
\algdef{SE}[CASE]{Case}{EndCase}[1]{\algorithmiccase\ #1}{\algorithmicend\     \algorithmiccase}%
\algtext*{EndSwitch}%
\algtext*{EndCase}%

\algdef{SE}[EVENT]{Event}{EndEvent}[1]{\textbf{upon event}\ #1\ \algorithmicdo}{\algorithmicend\ \textbf{event}}%
\algtext*{EndEvent}

\newacronym{lfd}{LFD}{Local Fault Diagnoser}

\newcommand{\set}[1]{\mathcal{#1}} % sets
\newcommand{\ie}{\textit{i.e.,~}} % i.e.,
\newcommand{\eg}{\textit{e.g.,~}} % i.e.,
\newcommand{\inneighbor}[1]{\set{N}_{#1}^{\texttt{in}}}
\newcommand{\outneighbor}[1]{\set{N}_{#1}^{\texttt{out}}}

\newcommand{\outdegree}[1]{d_{#1}^{\texttt{out}}}

\usepackage{pifont}
\definecolor{qual1}{HTML}{7FC97F} % green
\definecolor{qual2}{HTML}{BEAED4} % lila
\definecolor{qual3}{HTML}{FD673A} % orange
\definecolor{qual4}{HTML}{386CB0} % blue
\definecolor{qual5}{HTML}{E31A1C} % red
\definecolor{Green}{HTML}{1F783A}
\definecolor{DarkRed}{RGB}{145,12,7} % EmakPalette1 - DarkRed

\newcommand{\cmark}{\textcolor{Green}{\ding{51}}}
\newcommand{\xmark}{\textcolor{DarkRed!50}{\ding{55}}}

\DeclareMathSymbol{\shortminus}{\mathbin}{AMSa}{"39}

\usepackage{graphicx}
\graphicspath{ {./images/} }

\usetikzlibrary{arrows,arrows.meta}

\newenvironment{list4}{
	\begin{list}{$\bullet$}{%
			\setlength{\itemsep}{0.05cm}
			\setlength{\labelsep}{0.2cm}
			\setlength{\labelwidth}{0.3cm}
			\setlength{\parsep}{0in} 
			\setlength{\parskip}{0in}
			\setlength{\topsep}{0in} 
			\setlength{\partopsep}{0in}
			\setlength{\leftmargin}{0.16in}}}
	{\end{list}}

\title{\Huge
Distributed Estimation and Control for LTI Systems \\under Finite-Time Agreement
}

\author{% <-this % stops a space
Camilla Fioravanti$^{1}$, Evagoras Makridis$^2$, Gabriele Oliva$^{1*}$, Maria Vrakopoulou$^{2}$, and Themistoklis Charalambous$^{2}$
\thanks{
$^1$Department of Engineering, University Campus Bio-Medico of Rome, Via Alvaro del Portillo, 21 - 00128 Roma, Italy. E-mails: $\{$c.fioravanti, g.oliva$\}$@unicampus.it.}
\thanks{$^2$Department of Electrical and Computer Engineering, School of Engineering, University of Cyprus, %1 Panepistimiou Avenue, 
2109 Aglantzia, Nicosia, Cyprus. E-mails: $\{$surname.name$\}$@ucy.ac.cy.}
\thanks{The project MINERVA has received funding from the European Research Council (ERC) under the European Union's Horizon 2022 research and innovation programme (Grant Agreement No. 101044629).}
\thanks{$^*$ corresponding author. }
}

\begin{document}

\maketitle
\thispagestyle{empty}
\pagestyle{empty}

%%%%%%%%%%%%%%%%%%%%%%%%%%%%%%%%%%%%%%%%%%%%%%%%%%%%%%%%%%%%%%%%%%%%%%%%%%%%%%%%

%%%%%%%%%%%%%%%%%%%%%%%%%%%%%%%%%%%%%%%%%%%%%%%%%%%%%%%%%%%%%%%%%%%%%%%%%%%%%%%%
\begin{abstract}
This paper considers a strongly connected network of agents, each capable of partially observing and controlling a discrete-time linear time-invariant (LTI) system that is jointly observable and controllable.
Additionally, agents collaborate to achieve a shared estimated state, computed as the average of their local state estimates. Recent studies suggest that increasing the number of average consensus steps between state estimation updates allows agents to choose from a wider range of state feedback controllers, thereby potentially enhancing control performance.
However, such approaches require that agents know the input matrices of all other nodes, and the selection of control gains is, in general, centralized. 
Motivated by the limitations of such approaches,  we propose a new technique where: (i) estimation and control gain design is fully distributed and finite-time, and (ii) agent coordination involves a finite-time exact average consensus algorithm, allowing arbitrary selection of estimation convergence rate despite the estimator's distributed nature. We verify our methodology's effectiveness using illustrative numerical simulations. 
\end{abstract}


%%%%%%%%%%%%%%%%%%%%%%%%%%%%%%%%%%%%%%%%%%%%%%%%%%%%%%%%%%%%%%%%%%%%%%%%%%%%%%%%%
\section{Introduction}
\IEEEPARstart{I}{n the last decades,} several works in the literature focused on solving the distributed estimation and control problem, mainly relying on the development of linear observers~\cite{rego2019distributed, Ugrinovskii:2013-conditions, zhu2014cooperative, park2012necessary,park2016design, mitra2018distributed,khajenejad2023distributed,liu2023distributed}.  
%\IEEEPARstart{I}{n the last decades,} several works in the literature focused on solving the distributed estimation and control problem (e.g., see~\cite{rego2019distributed, fioravanti2022towards} and references therein). 
%A large variety of approaches for distributed estimation and control is based on the development of linear observers.
In~\cite{Ugrinovskii:2013-conditions}, the authors establish the necessary conditions for observability in the dependence of the plant model and the underlying communication graph, while in~\cite{zhu2014cooperative} the existence of a distributed observer is guaranteed under the detectability and connectivity conditions. 
In~\cite{park2012necessary,park2016design}, a state-augmented observer is proposed, and necessary and sufficient conditions for their stabilizability are derived. 
More recently, in~\cite{mitra2018distributed}, Kalman's canonical observational decomposition to a multi-sensor environment is introduced. %the authors solved the distributed estimation problem of an LTI system by introducing Kalman's canonical observational decomposition to a multi-sensor environment. 
\begin{table}[h!]
\centering
\caption{\textcolor{black}{Comparison of this work against the state of the art.}}
%Asymptotic consensus methods need to terminate after a pre-set number of steps, while finite-time methods  terminate by themselves after a finite number of steps. The selection of estimation and control gains can be either centralized, uncoordinated, or coordinated.
{\tiny
\noindent \renewcommand{\arraystretch}{1}
\renewcommand{\tabcolsep}{1pt}
\resizebox{0.49\textwidth}{!}{
\begin{tabular}{|l|c|c|c|c|}  
    \cline{1-5}
    %\hline
    \hline
    Approach & Discrete& Consensus & Selection of &  Selection of\\
     & Time& Consensus & Estimation Gains &  Control Gains \\
    \hline
    \hline
    \cite{kim2016distributed} - 2016
  & \xmark & \xmark & uncoordinated & n/a \\
    \hline
    \cite{park2016design} - 2017& \cmark & \xmark & centralized & n/a \\
    \hline
    \cite{han2018simple} - 2019 & \xmark & \xmark & uncoordinated & n/a \\
    \hline
    \cite{wang2019distributed} - 2019 & \xmark & \xmark & uncoordinated & n/a \\
    \hline
    \cite{wang2019distributedb} - 2019 & \cmark & asymptotic & uncoordinated & n/a \\
    \hline
    \cite{wang2020distributed} - 2020& \xmark & \xmark & \xmark & \xmark \\
    \hline
    \cite{savas2022separation} - 2022& \cmark & asymptotic & centralized & uncoordinated \\
    \hline
    {\textbf{This work}} & \cmark & \textbf{finite-time} & \textbf{coordinated} & \textbf{coordinated} \\
    \hline
    %\hline
\end{tabular}
}}
\label{tab:distributed_opt_algorithms}
\vspace{-15pt}
\end{table}

Although all the aforementioned works provided suitable solutions to the distributed estimation problem, none of them was able to directly control the convergence rate. This issue is overcome in~\cite{kim2016distributed,han2018simple, kim2019completely}, which developed the idea of the observability decomposition for the construction of a decomposed diffusive coupling. Specifically, in~\cite{han2018simple} the authors exploit linear matrix inequality (LMI) techniques to compute the observer gains, while in~\cite{kim2019completely} the distributed observer design is improved by eliminating the dependence on the communication topology among the observers.
In~\cite{wang2019distributedb}, the estimation error of discrete-time, jointly observable LTI system whose observations are distributed across a time-varying network, converges exponentially fast to zero at a fixed but arbitrarily chosen rate, assuming that the graph is strongly connected. Specifically, the agents are allowed to perform a given number of rounds of distributed agreement in between state updates.
The continuous-time counterpart was presented in~\cite{wang2019distributed}.
In \cite{wang2020distributed}, the authors proposed a distributed observer-based control system where each agent can independently adjust the controlled output to any desired set point value for a continuous-time system. However, although the solutions presented in \cite{han2018simple, kim2019completely, kim2016distributed, wang2019distributedb, wang2019distributed, wang2020distributed} succeed in the convergence rate assignment and can be implemented in a distributed manner, they all always require a degree of centralized coordination for the design of the estimation gains.
In Table \ref{tab:distributed_opt_algorithms}, we analyze the significant aspects of distributed gains design shown by the aforementioned works. 
Interestingly, only the approaches in \cite{wang2020distributed,savas2022separation} and in this paper also account for the control of the system, while all other approaches only focus on state estimation. Indeed, our approach is the only one that allows selecting the local controller gains.
%\textcolor{cyan}{\st{Moreover, other approaches at the state of the art may be unsuitable for selecting the estimation gains when finite-time consensus is performed between estimation steps.
In contrast to these approaches, in our paper, we aim to develop a strategy where agents' local estimation and control gains are designed and executed in a fully distributed manner.
\textcolor{black}{In particular}, we took inspiration from~\cite{savas2022separation}, in which a distributed estimation and control scheme for LTI systems is developed whereby agents are allowed to perform a given number of asymptotic average consensus iterations between successive time updates.
Additionally, the authors show how the frequency of information exchange between neighboring agents (consensus steps) affects the performance of the system since it affects the set of admissible controllers. Note that, also in this case, the control and estimation gains are known \emph{a priori} to the agents (\eg they are computed during a centralized design phase).

\color{black}
The main contributions of this work, which address the aforementioned limitations, are: (i) a novel fully distributed method for designing estimation and control gains through coordination, which unlike other approaches, enables self-configuration and resilience to node failures in multi-agent systems;
(ii) an exact finite-time average consensus mechanism (of precomputed required iterations) between estimation steps for precise eigenvalue placement, ensuring estimation error convergence to zero, offering stability guarantees which are not provided by other asymptotic methods where the consensus iterations may be insufficient.
These innovations set our approach apart, particularly in coordinated estimation and control gain selection, and in providing stability assurances. Our simulations demonstrate enhanced stability and performance in communication and computation time and rounds.
\color{black}

\section{Preliminaries}\label{sec:preliminaries}

\subsection{Notation and Graph Theory}
We denote vectors with boldface lowercase letters and matrices with uppercase letters. The transpose of matrix $A$ and vector $\mathbf{x}$ are denoted as $A^\top$, $x^\top$, respectively.  We refer to the \mbox{$(i,j)$-th} entry of a matrix $A$ by $A_{ij}$.
We represent by ${\bm 0}_n$ and ${\bm 1}_n$ vectors with $n$ {entries}, all equal to zero and to one, respectively.
We use  $\|\cdot\|$ to denote the Euclidean norm.

Let $\mathcal{G}=\{\mathcal{V},\mathcal{E}\}$ be a directed graph (digraph) with $N$ nodes \mbox{$\mathcal{V}=\{ v_1, v_2, \ldots, v_N \}$} and $e$ edges $\mathcal{E}\subseteq \mathcal{V}\times \mathcal{V}$, where $(v_i,v_j)\in \mathcal{E}$ captures the existence of a link from node $v_i$ to node $v_j$. 
A directed graph is {\em strongly connected} if each node can be reached by every other node via the edges, respecting their orientation.
Let the in-neighborhood $\mathcal{N}^{\texttt{in}}_i$ of a node $v_i\in\mathcal{V}$ be the set of nodes $v_j\in\mathcal{V}$ such that $(v_j,v_i)\in \mathcal{E}$; similarly, the out-neighborhood $\mathcal{N}^{\texttt{out}}_i$ of a node $v_i\in\mathcal{V}$ is the set of nodes $v_j\in\mathcal{V}$ such that $(v_i,v_j)\in \mathcal{E}$. In the case $\mathcal{N}^{\texttt{in}}_i$ is equal to $\mathcal{N}^{\texttt{out}}_i$ the graph is said to be \emph{undirected} and the neighborhood of a node $v_i\in\mathcal{V}$ is expressed with $\mathcal{N}_i$.
The {\em in-degree} $d_i^{\texttt{in}}$ of a node $v_i$ is the number of its incoming edges, \ie \mbox{$d_i^{\texttt{in}} = |\mathcal{N}_i^{\texttt{in}}|$};  similarly, the {\em out-degree} $d_i^{\texttt{out}}$ of a node $v_i$ is the number of its outgoing edges, \ie \mbox{$d_i^{\texttt{out}} = |\mathcal{N}_i^{\texttt{out}}|$}.
%In the following, where needed, we associate a weight $w_{ij}$ to each edge $(v_i,v_j)\in \mathcal{E}$, and we use $w_{ii}$ to model a weight associated to the agent $i$ itself. 

\subsection{Average (Minimum-Time) Consensus in Directed Graphs}
\label{sec:RatioConsensus}

The problem of average consensus over digraphs involves a number of nodes in a network represented by a graph $\set{G}$, that exchange information to compute the network-wide average of their initial values. 
We consider a synchronous setting in which each
node $v_j$ updates and sends its information to its neighbors at discrete times $t_{0}, t_{1}, t_{2}, \ldots$ (herein called \emph{consensus iterations}). We use ${\bm \alpha}_j[m] \in \mathbb{R}^n$ to denote the information state of node $v_j$ at time $t_m$. We index nodes' information states and any other information at time $t_m$ by $m$. At each consensus iteration $m$, each node $v_j$ maintains its information state ${\bm \alpha}_j[m] \in \mathbb{R}^n$ (initialized at a certain value, \eg a measurement, say ${\bm \alpha}_j[0]={\bm \theta}_j$), an auxiliary scalar variable $\pi_j[m] \in \mathbb{R}_+$ (initialized at $\pi_j[0]=1$), and ${\bm \mu}_j[m]:={\bm \alpha}_j[m]/\pi_j[m]$. \textcolor{black}{The authors in \cite{2010Christoforos:RC} showed that the nodes in an unbalanced digraph can asymptotically reach the exact average consensus by having each node $v_j \in \set{V}$ to iteratively communicate and update its values according to the so called \emph{ratio consensus} algorithm:}
\begin{subequations}\label{eq:ratio_consensus}
  \begin{align}
        {\bm \alpha}_j[m+1] &= p_{jj} {\bm \alpha}_j[m] + \sum_{v_i\in \inneighbor{j}} p_{ji} {\bm \alpha}_i[m], \label{eq:auxiliary_x} \\
        \pi_j[m+1] &= p_{jj} \pi_j[m] + \sum_{v_i\in \inneighbor{j}} p_{ji} \pi_i[m],\label{eq:auxiliary_y}
   \end{align}
\end{subequations}
\textcolor{black}{where $p_{lj}$ is a nonnegative weight assigned to outgoing information by each agent $v_j$ as $p_{lj}=1/(1+\outdegree{j})$ if $v_l \in \outneighbor{j} \cup \{v_j\}$, otherwise $p_{lj}=0$. The weights assigned based on this strategy form a column-stochastic matrix $P=\{p_{ji}\} \in \mathbb{R}_{+}^{N \times N}$ in which (possible) zero-valued entries, $p_{ji}=0$, denote the absence of a (directed) link $(v_j,v_i)$.}
% \begin{align}\label{eq:weights}
%     p_{lj}=\begin{cases}
%     \dfrac{1}{1 + \outdegree{j}}, & v_l \in \outneighbor{j} \cup \{v_j\},\\
%      0, & \text { otherwise},
%     \end{cases}
% \end{align}
%
Then, the solution to the average consensus problem can be asymptotically obtained. In particular, we have that~\cite{2010Christoforos:RC}
$$
\lim _{m \rightarrow \infty} {\bm \mu}_j[m]=\frac{\sum_{v_j \in \mathcal{V}} {\bm \alpha}_j[0]}{\sum_{v_j \in \mathcal{V}} \pi_j[0]}=\frac{\sum_{v_j \in \mathcal{V}} {\bm \theta}_j}{N},~ \forall v_j \in \mathcal{V}.
$$

%Note that the ratio consensus algorithm establishes that the exact average is \emph{asymptotically} reached, even if the directed graph is not balanced.

%\subsection{Minimum-Time Average Consensus}
%\label{sec:finiteaverage}

The results in \cite{charalambous2015distributed} have shown that the exact average ${\bm \mu}$ can be distributively obtained in a finite and minimum number of steps in strongly connected digraphs. The results are based on the use of the concept of the minimal polynomial associated with the linear dynamics of each of \eqref{eq:auxiliary_x} and \eqref{eq:auxiliary_y}, in conjunction with the final value theorem, initially proposed for undirected graphs in \cite{2013:Ye}.

Since the update of each element $r$, $r\leq n$, of the information state of node $v_j$ at time instant $m$, denoted by ${\alpha}_{jr}[m] \in \mathbb{R}$, is independent of the update of all the other elements, we will explain how the minimum-time average consensus works for a scalar information state, for simplicity.
Consider the vectors of the differences between $2m+1$ successive discrete-time values at node $v_j$
%\begin{align*}
$\overline{\alpha}^{\top}_{2m} = (\overline{\alpha}_{jr}[0], \overline{\alpha}_{jr}[1], \ldots, \overline{\alpha}_{jr}[2m] )$ and
$\overline{\pi}^{\top}_{2m} = (\overline{\pi}_j[0], \overline{\pi}_j[1], \ldots, \overline{\pi}_j[2m] )$,
%\end{align*}
where $\overline{\alpha}_{jr}[l] = {\alpha}_{jr}[l+1]-{\alpha}_{jr}[l]$, and $\overline{\pi}_{j}[l] = {\pi}_j[l+1]-{\pi}_j[l]$ for $l=0,\ldots,2m$. 
%for the two iterations ${\alpha}_{jr}[m]$ and ${\pi}_j[m]$ at node $v_j$ (as given in \eqref{eq:auxiliary_x} and \eqref{eq:auxiliary_y}), respectively. 
Then, each agent $v_j$ constructs the Hankel matrices $\Gamma\{\overline{\alpha}^{\top}_{2m}\}$ and $\Gamma\{\overline{\pi}^{\top}_{2m}\}$. The construction for $\Gamma\{\overline{\alpha}^{\top}_{2m}\}$ (and respectively for $\Gamma\{\overline{\pi}^{\top}_{2m}\}$) is as follows:
\begin{align*}
\Gamma\{\overline{\alpha}^{\top}_{2m}\} \triangleq
\begin{bmatrix}
\overline{\alpha}_{jr}[0] & \overline{\alpha}_{jr}[1] & \ldots & \overline{\alpha}_{jr}[m] \\
\overline{\alpha}_{jr}[1] & \overline{\alpha}_{jr}[2] & \ldots & \overline{\alpha}_{jr}[m+1] \\
\vdots & \vdots & \ddots & \vdots \\
\overline{\alpha}_{jr}[m] & \overline{\alpha}_{jr}[m+1] & \ldots & \overline{\alpha}_{jr}[2m]
\end{bmatrix}.
\end{align*}
It has been shown in \cite{2013:Ye} that after $2(M_j +1)$ the kernel of the Hankel matrices $\Gamma\{\overline{\alpha}^{\top}_{2m}\}$ and $\Gamma\{\overline{\pi}^{\top}_{2m}\}$ for arbitrary initial conditions ${\bm \alpha}[0]$ and ${\bm \pi}[0]$ except a set of initial conditions with Lebesgue measure zero
%\footnote{A subset of $\mathbb{R}^n$ is said to be Lebesgue measure zero if, for every $\varepsilon > 0$, it can be covered with countably many products of $n$ intervals whose total volume is at most $\varepsilon$. All countable sets and all the subsets of $\mathbb{R}^n$ whose dimension is smaller than $n$ have Lebesgue measure zero in $\mathbb{R}^n$.}
become defective. Once they become defective, node $v_j$ computes $M_j$ among others and has all the information needed to compute the final value of the overall system and, hence, compute the exact average in a finite number of steps. 
For more details, see~\cite{charalambous2015distributed}.

\subsection{Distributed Termination and Synchronization}
\label{sec:termination}

In~\cite{themis:2018ECC_termination} a distributed termination mechanism is proposed to enhance an existing finite-time distributed algorithm to allow the nodes to agree on when to terminate their iterations, provided they have all computed their exact average. Specifically, their proposed method is based on the fact that the finite-time consensus algorithm in \cite{charalambous2015distributed} allows nodes in the network running iterations \eqref{eq:auxiliary_x} and \eqref{eq:auxiliary_y} to compute an upper bound of their \emph{eccentricity\footnote{The eccentricity of a node $v_j \in \set{V}$ is the length of the maximum distance between $v_j$ and any other node $v_i \in \set{V}$.}} and use this information for deciding when to terminate the process. The procedure is as follows: 

\begin{list4}
\item Once iterations \eqref{eq:auxiliary_x} and \eqref{eq:auxiliary_y} are initiated, each node $v_j$ also initiates two counters $c_j$, $c_j[0]=0$, and $r_j$, $r_j[0]=0$. Counter $c_j$ increments by one at every time step, \ie $c_j[m+1]= c_j[m]+1$. \textcolor{black}{Counter $r_j$ updates are described next.} %The way counter $r_j$ updates is described next.
\item Alongside iterations \eqref{eq:auxiliary_x} and \eqref{eq:auxiliary_y} a $\max$-consensus algorithm~\cite{2008:Cortes} is initiated as well, given by
\begin{align}\label{eq:maxconsensus}
\phi_j[m+1] = \max_{v_i \in \inneighbor{j} \cup \{v_j\}}\big\{ \max\{\phi_i[m],c_i[m]\} \big\}, 
\end{align}
with $\phi_j[0]=0$. Every time step $m$ for which $\phi_j[m+1]={\phi}_j[m]$, $r_j$ increments by one, but if, however, at any step $m'$, $\phi_j[m'+1] \neq \phi_j[m']$, then $r_j$ is set to zero, \ie 
\begin{align}\label{eq:rj}
r_j[m+1]=
\begin{cases}
0, & \text{if } \phi_j[m+1] \neq \phi_j[m], \\
r_j[m]+1, & \text{otherwise}.
\end{cases}
\end{align}
\item Once the square Hankel matrices $\Gamma\{\overline{\alpha}^{\top}_{M_j}\}$ and $\Gamma\{\overline{\pi}^{\top}_{M_j}\}$ for node $v_j$ lose rank, node $v_j$ saves the count of the counter $c_j$ at that time step, denoted by $m^o_j$, as $c^o_j$, \ie $c^o_j\triangleq c_j[m^o_j]$, and it stops incrementing the counter, \ie $\forall m'\geq m^o_j, c[m']=c_j[m^o_j]=c^o_j$. Note that $c^o_j=2(M_j+1)$. 
\item Node $v_j$ terminates \eqref{eq:auxiliary_x} and \eqref{eq:auxiliary_y} when $r_j$ reaches $c^o_j$.
\end{list4}
The main idea of this approach is that $2(M_j +1)$ serves as an upper bound on the maximum distance of any other node to node $v_j$. This quantity is not known initially but becomes known to node $v_j$ through the finite-time consensus algorithm \cite{charalambous2015distributed}. Therefore, the algorithm in~\cite{themis:2018ECC_termination} is used for computing distributively the number of steps of the consensus iterations $\overline{m}$ needed by all nodes to complete the finite-time consensus; specifically, \textcolor{black}{we have}
\mbox{$\overline{m} = 2\max_{v_j\in \mathcal{N}} \{2(M_j+1)\} -1.$}
Hence, after $\overline{m}$ steps, each node can terminate the iterations.
\begin{remark}
\label{rem:diam}
Note that, as a byproduct of the distributed termination mechanism, an upper bound on the diameter of the network, $D^{\prime}=\max_{v_j\in \mathcal{N}} M_j$, can be acquired once a node has completed the execution of the termination mechanism. 
\end{remark}

\section{Problem Statement}\label{sec:problemstatement}

The problem considered in this paper is based on the setup used in~\cite{savas2022separation, wang2020distributed} in which a network of agents aim at estimating and controlling a discrete-time LTI system, characterized by 
\begin{equation}
\label{eq:sysdyn}
\begin{cases}
    {\bm x}[k+1]=A{\bm x}[k]+\sum_{j=1}^N B_j{\bm u}_j[k], ~ {\bm x}[0]\in \mathbb{R}^n , \\
    {\bm y}_i[k]=C_i{\bm x}[k],\quad \forall v_i\in \mathcal{V},
\end{cases}
\end{equation}
where $A\in\mathbb{R}^{n\times n}$, $B_j\in\mathbb{R}^{n\times q_j}$, $C_i\in\mathbb{R}^{p_i\times n}$, ${\bm x}[k]\in\mathbb{R}^n$ is the state vector, and ${\bm u}_j[k]\in\mathbb{R}^{q_j}, {\bm y}_i[k]\in\mathbb{R}^{p_i}$ are, respectively, the $j-$th control input and the $i-$th system output at discrete times $t_k$ (information states and any other information at time $t_k$ is indexed by $k$), herein called \emph{estimation-control iterations}.

Let us consider $q=\sum_{i=1}^N q_i$, $p=\sum_{i=1}^N p_i,$ and let us define 
$B=\begin{bmatrix}
    B_1, \dots, B_N
\end{bmatrix}\in\mathbb{R}^{n\times q}$
and 
\mbox{$C= \begin{bmatrix} C_1^\top, \dots, C_N^\top\end{bmatrix}^\top\in\mathbb{R}^{p \times n}.$}
In the remainder of the paper, we make the following assumption.
%\smallskip
\begin{assumption}
\label{ass:1}
The system in Eq. \eqref{eq:sysdyn} is jointly controllable and jointly observable, \ie the pair $(A,B)$ is controllable and the pair $(A,C)$ is observable.
\end{assumption}
%\smallskip
\noindent Assumption~\ref{ass:1} not only implies that each agent can only partially observe the LTI system output and only partially control the LTI system state, but that there exist controllers and observers such that the system is stable and the estimation error converges to zero. Hence, even though each agent might not be able to fully control or observe the system on its own, all the nodes can do so collaboratively by exchanging information over the strongly connected graph $\mathcal{G}$. 

The problem is to design the estimation dynamics and feedback
control law for each agent, given its own set of information,
so that jointly the agents stabilize the LTI system. Towards this end, Savas \emph{et al.}~\cite{savas2022separation} proposed a distributed state estimation scheme with the following structure:
%\textcolor{red}{Qui non \'e correttissimo: fanno consensus su $\hat{\bm x}_i[k]$, chiamiamo il risultato (che \'e diverso per tutti perche non stanno facendo consensus tempo finito) $\overline{\bm x}_i[k]$. Per cui il membro di destra dell'equazione qui sotto deve essere con $\overline{\bm x}_i[k]$ e non con $\hat{\bm x}_i[k]$. Inoltre forse va detto prima dell'equazione dei round di consensus e non dopo. Parlo di correggere solo la roba che parla di quello che fa Savva, il nostro essendo uguale per tutti \'e $\overline{\bm x}[k]$ ed \'e corretto. Inoltre Savva lavora coi grafi indiretti, per cui se l'abbiamo definito metterei $n\mathcal{N}_i$ invece di $n\mathcal{N}_i^{\texttt{in}}$ nella formula sotto.}
\begin{align}
\label{eq:estimation}
    \hat{\bm x}_i[k+1]=&A\tilde{\bm x}_i[k]+L_i({\bm y}_i[k]-C_i\tilde{\bm x}_i[k])\\
    &+\sum_{j=1}^N B_jK_j\tilde{\bm x}_i[k]+\sum_{j\in\mathcal{N}_i} W_{ij}(\tilde{\bm x}_j[k]-\tilde{\bm x}_i[k]), \nonumber
\end{align}
where $\tilde{\bm x}_i[k]$ is the result of $m$-steps of consensus over the graph topology, using the state estimates $\hat{\bm x}_i[k]$ as the initial conditions; $L_i\in\mathbb{R}^{n\times p_i}$, $W_{ij}\in\mathbb{R}^{n\times n}$ are suitable gains that make the estimation error dynamics asymptotically stable, while the control input is chosen to be in the form
${\bm u}_i=K_i\tilde{\bm x}_i[k],$ with $K_i\in\mathbb{R}^{q_i\times n}$.
%\begin{align}
%\label{eq:estimation}
%    \hat{\bm x}_i[k+1]=&A\hat{\bm x}_i[k]+L_i({\bm y}_i[k]-C_i\hat{\bm x}_i[k])\\
%    &+\sum_{j=1}^N B_jK_j\hat{\bm x}_i[k]+\sum_{j\in\mathcal{N}_i^{\texttt{in}}} W_{ij}(\hat{\bm x}_j[k]-\hat{\bm x}_i[k]), \nonumber
%\end{align}
%where $L_i\in\mathbb{R}^{n\times p_i}$, $W_{ij}\in\mathbb{R}^{n\times n}$ are suitable gains that make the estimation error dynamics asymptotically stable., while the control input is chosen to be in the form
%${\bm u}_i=K_i\hat{\bm x}_i[k],$ with $K_i\in\mathbb{R}^{q_i\times n}$ \textcolor{red}{(anche qui ${\bm u}_i=K_i\overline{\bm x}_i[k],$ with $K_i\in\mathbb{R}^{q_i\times n}$)}.
%
Note that in~\cite{savas2022separation} matrices $K_i$ are designed offline, in a centralized fashion, so that  $A+\sum_{i=1}^N B_i K_i$ is Schur stable.
Moreover, the authors consider a scenario where the agents are able to run a certain number $m$ of consensus iterations via a linear consensus algorithm for balanced graphs on their estimated states, between two consecutive estimation-control iterations. 
Specifically, in~\cite{savas2022separation} it is shown that, for sufficiently large $m$, there are gains $L_i, W_{ij}$ which stabilize the overall control and estimation process (\ie there is a separation between estimation and control).

In this paper, we assume that the agents execute, instead, a  {\em finite-time consensus} procedure (as described in~Sections~\ref{sec:RatioConsensus} and~\ref{sec:termination}), between consecutive estimation-control iterations. 
Under this assumption, we show that the distributed state estimation process is simplified, enabling the agents to arbitrarily place the eigenvalues for the estimation error dynamics, but also allowing for a larger set of controllers, if implemented under the framework of~\cite{savas2022separation}. 
Given the existence of this set of controllers, we develop an initialization procedure that allows the agents to coordinate in order to choose appropriate gains $K_i$ and $L_i$ in a completely distributed way and in finite time.

\section{Finite-time Average Consensus as an Enabler for Distributed Estimation and Control}\label{sec:finite-time}
In this section, we propose the adoption of the finite-time exact ratio consensus iterations discussed in Section~\ref{sec:RatioConsensus}, which operate in between two successive estimation-control time steps as shown in Fig.~\ref{fig:steps}. Under this scheme, the local state estimate $\hat{{\bm x}}_i[k]$ of each node $v_i$ is fed into the consensus protocol that exchanges this information with its neighboring nodes, as detailed in Sections~\ref{sec:RatioConsensus} and~\ref{sec:termination}. 
Under Assumption~\ref{assum:numberofsteps}, after of at most $\overline{m}$ consensus iterations, each node $v_i$ has terminated the finite-time consensus procedure, and computed the average of the network-wide state estimate in $M_i$ consensus iterations. Hence, each node $v_i$ at time step $k$ knows the common state estimate that is the average of the network-wide state estimate, $\overline{\bm x}[k]$, \ie 
$\overline{\bm x}[k]=\frac{1}{N}\sum_{l=1}^N \hat{\bm x}_l[k].$

\begin{assumption}
\label{assum:numberofsteps}
At least $\overline{m}$ consensus steps can be completed between two successive estimation-control time steps.
\end{assumption}
%states that for being able to obtain $\overline{\bm x}[k]$ through the adopted finite-time consensus mechanism, we need to ensure that $\overline{m}$ consensus steps can be completed within two successive estimation-control time steps, as shown in Fig.~\ref{fig:steps}.
\begin{figure}[h]
    \centering
    \resizebox{0.99\columnwidth}{!}{
    \vspace{-20pt}
    \input{finite-time-consensus-iterations}
    }
    \vspace{-10pt}
    \caption{Finite-time average consensus iterations $m$ (blue ticks) within distributed estimation and control iterations $k$ (red ticks).}
    \label{fig:steps}
\end{figure}
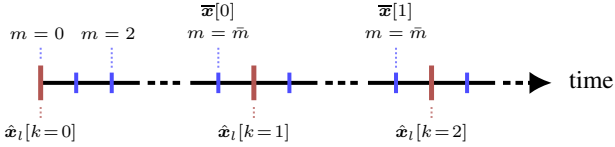

Notice that, due to such a procedure, the agents are able to reach an agreement on the estimate $\overline{\bm x}[k]$ of the state at time $k$; as a result, the systems' dynamics can be equivalently modeled as 
\begin{equation}
\label{eq:pippo}
    {\bm x}[k+1]=A{\bm x}[k]+\sum_{j=1}^N B_j K_j \overline{\bm x}[k]
\end{equation}
while since $\hat{\bm x}_i[k]-\hat{\bm x}_j[k]=\overline{\bm x}[k]-\overline{\bm x}[k]={\bm 0}_n$, the local estimator in Eq.~\eqref{eq:estimation} reduces to
\begin{align}\label{eq:estimationfinite}
    \hat{\bm x}_i[k+1]=&A\overline{\bm x}[k]+L_i\left({\bm y}_i[k]-C_i\overline{\bm x}[k]\right)+\sum_{j=1}^N B_jK_j\overline{\bm x}[k].\nonumber
\end{align}

We define the estimation error for the $i$-th agent as
\mbox{${\bm e}_i[k]={\bm x}[k]-\hat{\bm x}_i[k],$}
and the average estimation error
$$
    \overline{\bm e}[k]=\frac{1}{N}\sum_{l=1}^N{\bm e}_i[k]=\frac{1}{N}\sum_{l=1}^N ({\bm x}[k]-\hat{\bm x}_l[k])={\bm x}[k]-\overline{\bm x}[k].
$$

An instance of our proposed distributed estimation and control scheme \textcolor{black}{(which extends the approach in~\cite{savas2022separation})} for the problem in Section III, considering a simple network, is shown in Fig.~\ref{fig:closed_loop}.
In particular, the operations undertaken by each node $v_i$ in-between state updates essentially amount to an estimation task (\textcolor{black}{represented by the symbol} $\texttt{E}_i$ in the figure), a finite-time agreement task (\textcolor{black}{represented by the symbol} $\texttt{A}_i$) and a control task (\textcolor{black}{represented by the symbol} $\texttt{C}_i$). The figure shows the relations among such tasks for each node.
Interestingly, as discussed next, the nodes, besides being coupled by the physical process, directly interact by exchanging information only within the finite-time agreement task, while the estimation and control tasks are performed locally, based on the agreed-upon estimate computed during the finite-time agreement procedures.
\color{black}
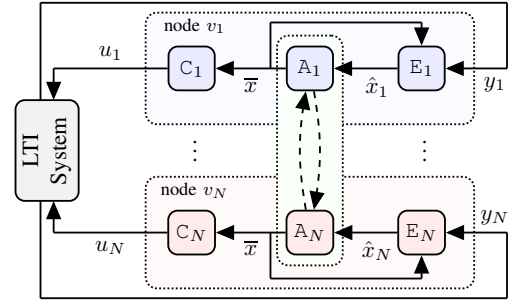
\begin{figure}[ht!]
    \vspace{-10pt}
    \centering
    \resizebox{0.8\columnwidth}{!}{
    \input{system_diagram} % closed_loop is the old one
    }
    \caption{\textcolor{black}{Proposed distributed estimation and control scheme.}}
    %For each node we show by $\texttt{E}_i$, $\texttt{A}_i$ and $\texttt{C}_i$ the operations undertaken by each node $v_i$ in between state updates, \ie estimation, finite-time agreement, and control.}}
    \label{fig:closed_loop}
    %\vspace{-10pt}
\end{figure}

In the next theorem, we show that finite-time consensus \textcolor{black}{in-between estimation steps}, allows arbitrary eigenvalue selection for the error dynamics\footnote{\textcolor{black}{Even if the agents execute finite-time average consensus in-between estimation steps, the overall estimation process converges asymptotically.}}.
%\smallskip
\begin{theorem}
Let Assumption~\ref{ass:1} hold. Then, the local estimation errors, ${\bm e}_i[k]$, can be brought to zero asymptotically, by selecting suitable gain matrices $L_i$ that assign arbitrary eigenvalues to 
$A-\frac{1}{N}\sum_{l=1}^N L_lC_l$.
\end{theorem}
%\smallskip
\begin{proof}
In order to prove the result, we observe that the dynamics of the network-wide average estimation error are described by
$$
    \begin{aligned}
    &\overline{\bm e}[k+1]\!=\!\frac{1}{N}\!\sum_{l=1}^N\!({\bm x}[k+1]\!-\!\hat{\bm x}_l[k+1])\!=\!{\bm x}[k+1]\!-\!\overline{\bm x}[k+1]\!\!\\
    &=\frac{1}{N}\sum_{l=1}^N
       \left(A{\bm x}[k]-A\overline{\bm x}[k]-L_l({\bm y}_l[k]-C_l\overline{\bm x}[k])\right)\\
    &=\frac{1}{N}\sum_{l=1}^N(A-L_lC_l)(\underbrace{{\bm x}[k]-\overline{\bm x}[k]}_{\overline{\bm e}[k]})=\left(A- \frac{1}{N} \mathcal{L}\,C\right)\overline{\bm e}[k],
    \end{aligned}
$$
where 
\mbox{$
\mathcal{L}=\begin{bmatrix}
    L_1, \dots, L_N
\end{bmatrix},
$}
and, thus, 
$\mathcal{L}\,C=\sum_{l=1}^NL_lC_l.$
Since, by Assumption \ref{ass:1}, the pair $(A,C)$ is observable, we have that also $(A,C/N)$ is observable and thus, by definition, there is a gain $\mathcal{L}$ which brings the error $\overline{\bm e}$ to zero by assigning arbitrary eigenvalues.

Let us now consider the dynamics of the local estimation errors ${\bm e}_i[k]$.
In particular, we have that
$$
\begin{aligned}
    {\bm e}_i[k+1] &= {\bm x}[k+1]-\hat{\bm x}_i[k+1]\\
    &= A{\bm x}[k]-A\overline{\bm x}[k]-L_l({\bm y}_l[k]-C_l\overline{\bm x}[k])\\
    &= \left(A-L_iC_i\right) \left({\bm x}[k]-\overline{\bm x}[k]\right)=\left(A-L_iC_i\right)\overline{\bm e}[k].
    \end{aligned}
$$
Therefore, being the local estimation errors ${\bm e}_i[k+1]$ essentially a scaling of $\overline{\bm e}[k]$ (for which we already proved convergence to zero) by a constant matrix $A-L_iC_i$, we have that also the terms ${\bm e}_i[k]$ converge to zero. This completes our proof.
\end{proof}
%The next remark provides evidence that previous distributed schemes for the selection of the observer gains \cite{kim2016distributed,han2018simple,wang2019distributed,wang2019distributedb} may not be adequate to choose the control gains $K_i$.
%\smallskip
\begin{remark}
In~\cite{kim2016distributed,han2018simple,wang2019distributed,wang2019distributedb} the gains $L_i$ are designed locally. 
However, due to the presence of multiple input signals, this strategy is not guaranteed to work for the design of the control gains $K_i$.
Consider, for instance, the case where $N=2$,
$$
A=\begin{small}
    \begin{bmatrix}
1&-2&0&0\\
0&-1&0&0\\
0.1&-0.1&0.5&0.1\\
0.2&-0.1&0.5&0.1
\end{bmatrix}
\end{small},
\quad B_1=\begin{small}\begin{bmatrix}
0\\1\\0\\1
\end{bmatrix}\end{small},\!\!
\quad
B_2=\begin{small}\begin{bmatrix}
1\\0\\1\\0
\end{bmatrix}\end{small}\!\!.\!\!
$$
Using the approach in~\cite{kim2016distributed,han2018simple,wang2019distributed,wang2019distributed} to choose $K_i$ we get 
\begin{align}
K_1&=\begin{bmatrix}
2.7788 & -2.0033 & 0.0436 & 1.5033
\end{bmatrix},\nonumber\\ 
K_2&=\begin{bmatrix}
-1.7909 & 4.0311 & -0.1091 & -5.0182
\end{bmatrix}.\nonumber
\end{align}
However, the eigenvalues of $A+B_1K_1+B_2K_2$ are unstable, being \mbox{$\{0.6515 \pm 2.8137\imath,-2.7102,0.5073
\}$}.
Notice that a similar problem would arise if the gains $L_i$ in our scheme were computed as in \cite{kim2016distributed,han2018simple,wang2019distributed,wang2019distributedb}, as $A-\frac{1}{N}\sum_{i=1}^\ell L_\ell C_\ell$ would not necessarily be stable.
\end{remark}
%\smallskip

\section{Distributed selection of the Gains}\label{sec:distributed_gain_computation}
In this section, we develop an initialization procedure that allows the agents to select suitable gains matrices $K_i$ and $L_i$ in finite time and in a completely distributed way.
Notice that, for the sake of brevity, in the remainder of this section we address the problem of choosing the control gains $K_i$; the application of the envisaged procedure to the selection of the observer gains $L_i$ is straightforward.

Specifically, the overall initialization procedure is comprised of three  consecutive initialization steps, described next:
\begin{itemize}
    \item[P1)] Run iterations of the finite-time average consensus algorithm twice, as in Sections~\ref{sec:RatioConsensus} and~\ref{sec:termination}, to obtain an upper bound $D'$ on the network diameter $D$ (see Remark~\ref{rem:diam}) and $\overline{m}$ \cite{themis:2018ECC_termination}. 
    \item[P2)] Perform a distributed leader election procedure, \eg a max-consensus procedure for $D'$ steps (computed already in procedure P1) over, \eg the agents' identifiers or by just selecting a number from a continuous fixed range of values~\cite{2008:Cortes,nejad2009max}.
    \item[P3)] Execute the procedure for the distributed selection of the gains discussed next, in Section~\ref{subsec:gainchoice}; it is based on the methodology discussed in Section~\ref{subsec:single}, which requires a leader, already elected in procedure P2.
\end{itemize}

\subsection{Single Eigenvalue Placement}
\label{subsec:single}
This subsection reviews a control technique for a system in the form ${\bm z}[k+1]=\mathcal{A}{\bm z}[k]+\mathcal{B}{\bm u}[k]$, with $\mathcal{A}\in\mathbb{R}^{n\times n}$ and $\mathcal{B}\in \mathbb{R}^{n\times q}$ that allows a single eigenvalue $\lambda_i$ of the dynamical matrix $\mathcal{A}$ to be modified while leaving all other eigenvalues unchanged.
In particular, in the following we assume that $\lambda_i$ is a controllable eigenvalue of $\mathcal{A}$, \ie 
\mbox{$
\texttt{rank}(\begin{bmatrix}\mathcal{A}-\lambda_i I_n & \mathcal{B}\end{bmatrix})=n
$}
as per the Popov-Belevitch-Hautus test (\eg see~\cite{luenberger1979introduction}). Moreover, let us consider a left eigenvector ${\bm w}_i(\mathcal{A})\in\mathbb{R}^n$ associated to $\lambda_i$, \ie  
\mbox{$
{\bm w}_i^\top(\mathcal{A})\mathcal{A}=\lambda_i {\bm w}_i^\top(\mathcal{A}).
$}
Let us now consider the dynamics \mbox{${\bm x}[k+1]=\mathcal{A}{\bm x}[k]+\mathcal{B}{\bm u}[k]$} and let us initially assume, for simplicity of exposition, that $\mathcal{B}\in \mathbb{R}^{n\times 1}$, \ie that ${\bm u}[k]$ is scalar.
Premultiplying the systems' dynamics  by ${\bm w}_i^\top(\mathcal{A})$, we highlight the dynamics along the eigenspace spanned by ${\bm w}_i^\top(\mathcal{A})$ (\eg see~\cite{luenberger1979introduction}), thus obtaining
$$
\begin{aligned}
{\bm w}_i^\top(\mathcal{A}) {\bm z}[k+1]&={\bm w}_i^\top(\mathcal{A})\mathcal{A}{\bm z}[k]+{\bm w}_i^\top(\mathcal{A})\mathcal{B} {u}[k]\\
& = \lambda_i{\bm w}_i^\top(\mathcal{A}){\bm z}[k]+{\bm w}_i^\top(\mathcal{A})\mathcal{B} {u}[k].
\end{aligned}
$$
Suppose that, as a result of the control action, we want to replace $\lambda_i$ with $\overline{\lambda}_i$, \ie we want to choose $u[k]$ such that 
$$
{\bm w}_i^\top(\mathcal{A}) {\bm z}[k+1]=\overline{\lambda}_i{\bm w}_i^\top(\mathcal{A}){\bm z}[k].
$$
This is possible by choosing $u[k]=\mathcal{K}{\bm z}[k]$ with
\begin{equation}
\label{eq:choicegainK}
  \mathcal{K}=\dfrac{\overline{\lambda}_i-\lambda_i}{{\bm w}_i^\top(\mathcal{A})\mathcal{B}} {\bm w}_i^\top(\mathcal{A}).  
\end{equation}
Notice that the above selection for $\mathcal{K}$ is well-posed; in fact, since $\lambda_i$ is controllable then $\begin{bmatrix}\mathcal{A}-\lambda_i I_n & \mathcal{B}\end{bmatrix}$ is full rank and thus  
\mbox{${\bm w}_i^\top(\mathcal{A})\begin{bmatrix}\mathcal{A}-\lambda_i I_n & \!\!\mathcal{B}\end{bmatrix}\neq 0,$}
from which we have that
$$
\begin{bmatrix}{\bm w}_i^\top(\mathcal{A})\mathcal{A}-\lambda_i {\bm w}_i^\top(\mathcal{A}) & \!\!{\bm w}_i^\top(\mathcal{A})\mathcal{B}\end{bmatrix}=\begin{bmatrix}0 &\!\! {\bm w}_i^\top(\mathcal{A})\mathcal{B}\end{bmatrix}\neq0,
$$
which implies that ${\bm w}_i^\top(\mathcal{A})\mathcal{B}\neq 0$.
The above $\mathcal{K}$ is such that $\mathcal{A}+\mathcal{B}\mathcal{K}$ has the same eigenvalues as $\mathcal{A}$, except for the fact that the eigenvalue $\lambda_i(\mathcal{A})$ is replaced by the desired $\overline{\lambda}_i$. Note that this is similar to a Hotelling deflation~\cite{mackey2008deflation}, since we only modify the direction spanned by one eigenvector, without modifying the other directions.

Notably, the procedure can be iterated. Assuming that we also want to replace $\lambda_j$ with $\overline{\lambda}_j$, it is sufficient to set
\mbox{$u[k]=\mathcal{K}{\bm z}[k]+r[k]$}, where $\mathcal{K}$ is as above, while $r[k]=\mathcal{H} {\bm z}[k]$, with
\begin{equation}
\label{eq:choicegainK2}
\mathcal{H}= \dfrac{\overline{\lambda}_j-\lambda_j}{{\bm w}_i^\top(\mathcal{A}+\mathcal{B}\mathcal{K})\mathcal{B}} {\bm w}_i^\top(\mathcal{A}+\mathcal{B}\mathcal{K}).
\end{equation}
Notice that the left eigenvector used to replace $\lambda_j$ is the one obtained after applying $K$, \ie considering the matrix $\mathcal{A}+\mathcal{B}K$.

Let us now consider the case where $\mathcal{K}\in\mathbb{R}^{n\times q}$ with $q>1$. In this case, assuming that a matrix $\mathcal{K}^{(j)}$ has been selected with the above approach for each column $\mathcal{K}^{(j)}$ of $B$, it is sufficient to set the overall $\mathcal{K}\in\mathbb{R}^{q\times n}$ as
%$$
%\mathcal{K}=\begin{bmatrix}\mathcal{K}^{(1)}\\ \vdots \\ \mathcal{K}^{(q)}\end{bmatrix}.
%$$
\mbox{$
\mathcal{K}=\begin{bmatrix}\mathcal{K}^{(1)\top}&\ldots & \mathcal{K}^{(q)\top}\end{bmatrix}^\top.
$}
\subsection{Proposed Protocol for Choosing the Gains}
\label{subsec:gainchoice}
% \begin{figure}
% \centering
% \resizebox{0.3\textwidth}{!}{
% \begin{tikzpicture}[
%             > = stealth, % arrow head style
%             shorten > = 1pt, % don't touch arrow head to node
%             auto,
%             node distance = 3cm, % distance between nodes
%             thick % line style
%         ]
% \begin{scope}[every node/.style={circle,thick,draw},
%               main/.style={draw=black, thick, draw, circle, fill=qual4!10, minimum size=0.9cm},]
% \node[main] (v1) {$v_i$};
% \node[main] (v2) [below left of=v1] {$v_j$};
% \node[main] (v3) [ below of=v1] {$v_h$};
% \node[main] (v4) [below right of=v1] {$v_\ell$};
% \end{scope}

% \begin{scope}[>={Stealth[black]},
%               %every node/.style={fill=white,circle,minimum size=0.1},
%               every edge/.style={draw=black, thick, ->,> = latex'}]
% %   FROM        BEND/LOOP           POSITION OF LABEL           LABEL           TO
% \path[->,very thick, sloped] (v1) edge[bend right] node {first} (v2);
% \path[->,very thick,sloped] (v1) edge node {second} (v3);
% \path[->,very thick,sloped] (v1) edge[bend left] node {last} (v4);
% \end{scope}
% \end{tikzpicture}
% }
% \caption{Example of round robin token transmission policy. Node $v_i$ decides priorities for sending the token to its neighbors (\eg to $v_j$ first, to $v_h$ if it holds the token a second time, and to $v_\ell$ last). Whenever $v_i$ has the token, it is transmitted in cyclic order according to such priorities. \textcolor{red}{[GAB: This figure is really unnecessary, we can safely remove it.]}}
%  \label{fig:roundrobin0}
% \end{figure}
The main idea of our protocol is to have the agents choose $K_i$ one after another, using the approach discussed in Section~\ref{subsec:single}.
In the beginning, a leader is elected in a distributed way\footnote{Leader election could be implemented via max-consensus over the agents' identifiers~\cite{olfati2004consensus,muniraju2019analysis}.}. The leader initiates a token passing procedure where a token is circulated among the agents, and the token contains a matrix $F$ to be described next.
%
%
%The main idea of the proposed approach is to let the agents elect a leader $v_i$ in a distributed way, and that the leader chooses a gain matrix $K_i$ which stabilizes all the eigenvalues that are controllable by the $i$-th agent alone, without affecting the uncontrollable ones; to this end, the agent will choose $K_i$ using the technique discussed in Section~\ref{subsec:single}.
%Note that the . 
%
%Now, the leader passes a message $F=B_1K_1$ to an out-neighbor $v_j$, following a round-robin approach\footnote{\EM{Node $v_i$ passes the token to its neighbors in cyclic order according to specified priorities.}}%, as shown in Fig.~\ref{fig:roundrobin0}.
%When agent $v_j$ receives the message, it selects $K_j$ that stabilizes all the eigenvalues which can be stabilized by $v_j$ alone, again using the technique in Section~\ref{subsec:single}, but considering the left eigenvectors of $ A+B_iK_i$;  then agent $v_j$ adds $B_jK_j$ to $F$ and sends the message
%\mbox{$F=B_iK_i+B_jK_j$} to another of its out-neighbors, following a round-robin perspective.

In general, when a node $v_j$ receives a message $F$ for the first time, it checks if $A+F$ is  Schur stable. If $A+F$ is not Schur stable, it computes a gain $K_j$ that stabilizes all the eigenvalues it can (as done in Section~\ref{subsec:single}), it adds its own term $B_jK_j$ to the message and it sends the updated message
$F+B_jK_j$ to an out-neighbor that has not yet been visited.
Notice that some of the agents might not be able to control any eigenvalue, and in that case, they simply choose $K_j=0_{n\times q_i}$.
If, conversely, upon receiving the message for the first time, agent $j$ realizes that $A+F$ is Schur stable, the message becomes read-only and starts being flooded in the network.
A last case is when a message that is not read-only is received by a node that has already received it; in this case, the message is simply forwarded to a not-yet-visited out-neighbor.
Notice that the proposed approach does not require a stopping criterion. In fact, since the system is jointly controllable by all agents, and due to the approach adopted for the selection of the gains, eventually, all unstable eigenvalues will be placed and one of the agents will thus receive $F$ such that $A+F$ is Schur stable; as a consequence, the read-only phase will begin. Also in this case, since the graph is strongly connected, and since each agent broadcasts the message to all its neighbors and then stops participating, the read-only message will reach all nodes in finite time.
Algorithm~\ref{alg:choiceK} summarizes the above procedure.
\begin{algorithm}
%\begin{footnotesize}
\caption{Choice of $K_i$ (from the perspective of agent $i$)}\label{alg:choiceK}
\begin{algorithmic}
\Procedure{Initialization}{$\,$}
\State Assign a priority to each out-neighbor;
\State Declare to be not yet visited;
\State $K_i\gets 0_{n\times q_i}$;
\EndProcedure
\Procedure{OnLeaderElected}{$\,$}
\State Create empty token $F=0$;
\State Trigger OnReceiveToken($F$);
\EndProcedure
\Procedure{OnReceiveToken}{$F$}

\eIf{already visited}{
	Forward $F$ to an out-neighbor according to the priorities;
 }{
	\eIf{$A+F$ is not Schur stable}{
		Compute $K_i$ that stabilizes controllable eigenvalues\textcolor{black}{\footnotemark};
	
  		$F\gets F+K_iB_i$;
	
 		Send $F$ to an out-neighbor according to the priorities;
	}{
		Declare $F$ read-only;
		
		Broadcast $F$ to all out-neighbors;
		
		Stop Participating;
	}
}
\EndProcedure

\Procedure{OnReceiveReadOnlyMessage}{$F$}
		\State Broadcast $F$ to all out-neighbors;
		
		\State Stop Participating;
\EndProcedure
\end{algorithmic}
%\end{footnotesize}
\end{algorithm}
\footnotetext{\textcolor{black}{The matrix $K_i$ can be computed using the approach in Section~\ref{subsec:single}, i.e., via  Eqs.\eqref{eq:choicegainK}--\eqref{eq:choicegainK2}.}}

The following theorem establishes the correctness of the proposed algorithm and characterizes the associated time and message complexity.

\begin{theorem}
Let  $\mathcal{G}=\{\mathcal{V},\mathcal{E}\}$ be strongly connected, let Assumption~\ref{ass:1} hold true, and assume that the algebraic and geometric multiplicities of the unstable eigenvalues of $A$ are the same.
Then, Algorithm~\ref{alg:choiceK} is {\em correct}, i.e., each agent computes $K_i$ such that $A+\sum_{\ell=1}^N B_iK_i$ is Schur stable. Moreover, the time complexity and the message complexity of Algorithm~\ref{alg:choiceK} amount to $\mathcal{O}(N^2)$. Finally, at each step, $\mathcal{O}(N)$ messages are transmitted.
\end{theorem}
\begin{proof}
In order to prove our statement, let us first show that, over a strongly connected graph $G$, the message will be received by every agent (hence each agent will add its contribution $B_iK_i$) in at most $\mathcal{O}(N^2)$ steps.

To this end, let us consider the worst case in terms of graph topology and priorities and let us use $v_i$ to denote the agents with the order they are first reached by the token (hence, the leader is $v_1$).
Moreover let $\omega_i$ denote the number of steps required to visit node $v_{i}$ from the time the token is first received by node  $v_{i-1}$ (in this view, $\omega_1=0$).
\textcolor{black}{Over a strongly connected graph, in the worst case, the token always goes back to $v_1$ after visiting a new node. Therefore, node $v_i$ is visited after $\omega_i=2i-3$ steps; in fact, it takes $i-2$ steps to go back to $v_1$ and additional $i-1$ steps to reach $v_i$ from $v_1$.}
Hence, the total message complexity is \mbox{$\sum_{i=2}^N \omega_i = \sum_{i=2}^N\left(2 i-3\right) = N^2\!-\!4N\!+\!2\!=\!\mathcal{O}(N^2).$}
%We established that the token passing procedure visits all nodes in $\mathcal{O}(N^2)$ steps.

At this point, we observe that since by Assumption~\ref{ass:1} the pair $(A,B)$ is controllable, and since the gain matrices $K_i$ are chosen iteratively, each time modifying only the eigenspace spanned by the eigenvalues that are unstable and controllable by the $i$-th agent alone, by construction $A+B\mathcal{K}$ is Schur stable, with \mbox{$\mathcal{K}=\begin{bmatrix}K_1^\top,\ldots,K_N^\top\end{bmatrix}^\top.$}
Correctness follows by noting that, when the last agent receives the message for the last time and updates $F$, it holds
\mbox{$A+B\mathcal{K}=A+\sum_{\ell=1}^N B_i K_i = A+F.$}
Notice that, since the token passing requires $\mathcal{O}(N^2)$ steps in the worst case, the message is declared as read-only in $\mathcal{O}(N^2)$ steps, and the read-only message is received by all agents in additional $\mathcal{O}(N)$ steps (\ie because the message is flooded over the strongly connected graph). As a consequence, the time complexity of Algorithm~\ref{alg:choiceK} is $\mathcal{O}(N^2)$.
Let us now discuss the message complexity. In particular, during the initial phase in which the nodes transmit the token on a point-to-point basis, $\mathcal{O}(N^2)$ messages are transmitted. Then, when the token becomes read-only, the message is transmitted once for each link, \ie $\mathcal{O}(|E|)$ times.
Therefore, overall, the message complexity of Algorithm~\ref{alg:choiceK} is \mbox{$\mathcal{O}(N^2+|E|)=O(N^2).$}
To conclude we observe that, in the worst case of a star topology where the star-center broadcasts the read-only message, $\mathcal{O}(N)$ messages are transmitted at each time step.
\end{proof}
%\smallskip
%A few remarks are now in order regarding Algorithm~\ref{alg:choiceK}.
%\smallskip
\begin{remark}
As a byproduct of Algorithm~\ref{alg:choiceK}, each node will know $\sum_{\ell=1}^N B_\ell K_\ell$, which can be used within Eq.~\eqref{eq:pippo}, thus relaxing the assumption in \cite{savas2022separation}, in which each agent $v_i$ is required to know $\sum_{j\in \mathcal{N}^{\texttt{in}}_i} B_jK_j$ \emph{a priori}.
It is worth mentioning that the input matrices $B_i$ are never directly shared; instead, the nodes accumulate $B_iK_i$ in the token.
\end{remark}
\begin{remark}
%\smallskip
Algorithm~\ref{alg:choiceK} can be applied also to the selection of the gains $L_i$. In fact, it is sufficient to apply the algorithm to the selection of gain matrices $\widetilde{K}_i$ with respect to the pairs $(A^\top, -\frac{1}{n}C_i^\top)$, and then set $L_i=\widetilde{K}_i^\top$.
\end{remark}
%\smallskip

%%%%%%%%%%%%%%%%%%%%%%%%%%%%%%%%%%%%%%%%%%%
\section{Illustrative Examples}\label{sec:example}

Consider a directed network consisting of four nodes, described by matrix $P$, where each node has access to the dynamical matrix $A$. 
$$\footnotesize \setlength\arraycolsep{2.2pt}
    P \!=\!\!\! \begin{bmatrix}
        1/3 & 0 & 1/4 & 1/3\\
        1/3 & 1/2 & 1/4 & 0\\
        0 & 1/2 & 1/4 & 1/3\\
        1/3 & 0 & 1/4 & 1/3
    \end{bmatrix}\!\!,
    A \! = \!\! \begin{bmatrix}
        1&0.5&0&0&3&0&0&0\\
        0.5&1&0&0&0&0&0&0\\
        0&0&1&0.5&0&0&0&0\\
        0&0&0.8&1&0&0&0&0\\
        0&0&0&0&1&0.5&0&0\\
        0&0&0&0&0.6&1&0&0\\
        0&0&0&0&0&0&0.7&0.1\\
        1&0&0&0&0&0&0.2&0.7
    \end{bmatrix}\!\!.\!\!    
$$
% \begin{figure}[h]
%     \centering
%     \resizebox{0.55\columnwidth}{!}{
%     \input{graph_simulations}
%     }
%     \caption{Digraph with four nodes, considered in the simulations in Section~\ref{sec:example}. Self-loops are omitted for brevity, although nodes have access to their own states.}
%     \label{fig:graph}
% \end{figure}
Moreover, each node in the network has sensing, and computational capabilities to run the finite-time exact average consensus algorithm (for this setup $\overline{m}=11$ steps) and compute its local state estimates, and control input signals. More specifically, each node can sense only one state, and apply only one control signal to the system. The local input matrices $B_i$ for each node $v_i$ are $B_1=\setlength\arraycolsep{2pt}
\begin{bmatrix}
    0 & 1 & 0 & 0 & 0 & 0 & 0 & 0
\end{bmatrix}^{\top}, B_2=\begin{bmatrix}
    0 & 0 & 0 & 1 & 0 & 0 & 0 & 0
\end{bmatrix}^{\top}, B_3=\begin{bmatrix}
    0 & 0 & 0 & 0 & 0 & 1 & 0 & 0
\end{bmatrix}^{\top}, \text{ and } B_4=\begin{bmatrix}
    0 & 0 & 0 & 0 & 0 & 0 & 0 & 1
\end{bmatrix}^{\top}
$. Similarly, the local measurement matrices $C_i$ are $
C_1=\setlength\arraycolsep{2pt}\begin{bmatrix}
1&0&0&0&0&0&0&0\end{bmatrix}, C_2=\begin{bmatrix}0&0&1&0&0&0&0&0\end{bmatrix}, C_3=\begin{bmatrix}0&0&0&0&1&0&0&0\end{bmatrix}, \text{ and } C_4=\begin{bmatrix}0&0&0&0&0&0&1&0\end{bmatrix}$.

%thus the columns of $B$ and rows of $C$, respectively, where

%\begin{align*}
%B&=\begin{bmatrix}
%0&1&0&0&0&0&0&0\\
%0&0&0&1&0&0&0&0\\
%0&0&0&0&0&1&0&0\\
%0&0&0&0&0&0&0&1
%\end{bmatrix}^\top,\\
%C&=\begin{bmatrix}
%1&0&0&0&0&0&0&0\\
%0&0&1&0&0&0&0&0\\
%0&0&0&0&1&0&0&0\\
%0&0&0&0&0&0&1&0
%\end{bmatrix}.
%\end{align*}

Let $\rho_i$ and $\chi_i$ denote, respectively, the controllability and observability indices for the $i$-th agent, \ie the rank of the controllability and observability matrices for the pairs $(A,B_i)$ and $(A,C_i)$, respectively. 
We have that
$\rho_1=4$, $\rho_2=2$, $\rho_3=6$, $\rho_4=2$, while $\chi_1=4$, $\chi_2=2$, $\chi_3=2$, $\chi_4=6$, \ie the system is not independently observable or controllable by the single agents. However, the system is jointly observable and controllable as the rank of both the overall controllability and observability matrices is full.

Let us now consider the initialization phase, and let us assume that $v_1$ is elected as leader and that the token visits the agents in order from $v_1$ to $v_4$. The agents choose $L_i$ and $K_i$ via Algorithm~\ref{alg:choiceK}, with the aim to place the observer eigenvalues at \mbox{$\{0.20, 0.21, \ldots, 0.27\}$} and the controller eigenvalues at \mbox{$\{0.60, 0.61, \ldots, 0.67\}$}. The local observer and control gains chosen by each agent $v_i$ correspond to the $i$-th row of $K$ and column of $L$, respectively, where
$$
\footnotesize{
K\!=\!\!
\setlength\arraycolsep{1.8pt}
    \begin{bmatrix}
   \shortminus1.0916 &   0&    0.0003      &   0\\
   \shortminus1.0114 &   0&   0.0003   &      0\\
   0&  \shortminus1.0664  & 0&         0\\
   0&  \shortminus0.7300 &0&         0\\
   31.9630&  0&  \shortminus0.8021    &     0\\
   32.3204 & 0&   \shortminus0.5786      &   0\\
   \shortminus0.0350 &   0 &    0&         0\\
   \shortminus0.0247 &   0&    0&         0
\end{bmatrix}^\top\!\!\!\!\!\!, 
L\!=\!\!
\begin{bmatrix}
   12.5600&    0&         0&    0\\
   \shortminus8.1221&   0&         0&   0\\
    0&    6.0400&         0&   0\\
    0&    7.7600&         0&   0\\
    7.0165&    0&         0&    0\\
    7.6708&    0&         0&    0\\
   \shortminus0.1299&    0&         0&    3.4800\\
   \shortminus5.7761&   0&         0&    8.3680\\
\end{bmatrix}\!.}
$$
Here it is worth mentioning that, although all agents participate in all procedures involved in estimating and stabilizing the states of the system, the control gain of $v_4$ and the observer gain of $v_3$ are zero. 
In the first case, this occurs since when agent $v_4$ receives the token the matrix already has the desired eigenvalues, while in the second case, this is due to the fact that agent $v_3$ is only able to change some of the eigenvalues that have already been changed, hence there is a need to resort to the observer gain chosen by $v_4$.

%\begin{remark}
%    In the case the local control gains $K_i$ are already known by each agent $v_i$, then one could replace the distributed gain computation protocol proposed in Section~\ref{sec:distributed_gain_computation}, by additional consensus iterations over $B_i K_i$ during the initialization phase, to ensure that a global estimate of $BK$ is available at the local estimators in Eq.~\eqref{eq:estimationfinite}.
%\end{remark}
%Fig.~\ref{fig:simulation_results} depicts the estimation error of each agent $v_i$, ${\bm e}_i[k]$, at each estimation-control iteration $k$, the average (network-wide) estimation error $\overline{\bm e}[k]$, and the state evolution of the closed-loop system $\bm{x}[k]$. It is clear to see that, both the estimation error at each agent and the states $\bm{x}[k]$ of the system in Eq.~\eqref{eq:sysdyn} converge to zero (since we consider a regulation problem) as time evolves. 
%\begin{figure}[h!]
%    \centering
%    \input{results_estimation_error}
%    \caption{Local and average network-wide estimation errors at each estimation-control iteration $k$.}
%    \label{fig:simulation_results}
%\end{figure}

\begin{figure}[h]
    \centering
    \begin{subfigure}[t]{.99\columnwidth}
        \centering
        \input{lastfigure01}
    \end{subfigure}
    \vspace{-5pt}
    \begin{subfigure}[t]{.99\columnwidth}
        \centering
        \input{lastfigure1}
    \end{subfigure}
    \vspace{-5pt}
    \begin{subfigure}[t]{.99\columnwidth}
        \centering
        \input{lastfigure10}
    \end{subfigure}  
    %\vspace{-10pt}
    \caption{\textcolor{black}{Comparison of the estimation error between the proposed approach (solid line, $\overline{m}=11$) and the one in \cite{savas2022separation} when $m\in\{6,11,22\}$, for different communication duration $\tau \in \{0.1, 1, 10\}$.}}
    \label{fig:simulation_results}
\end{figure}
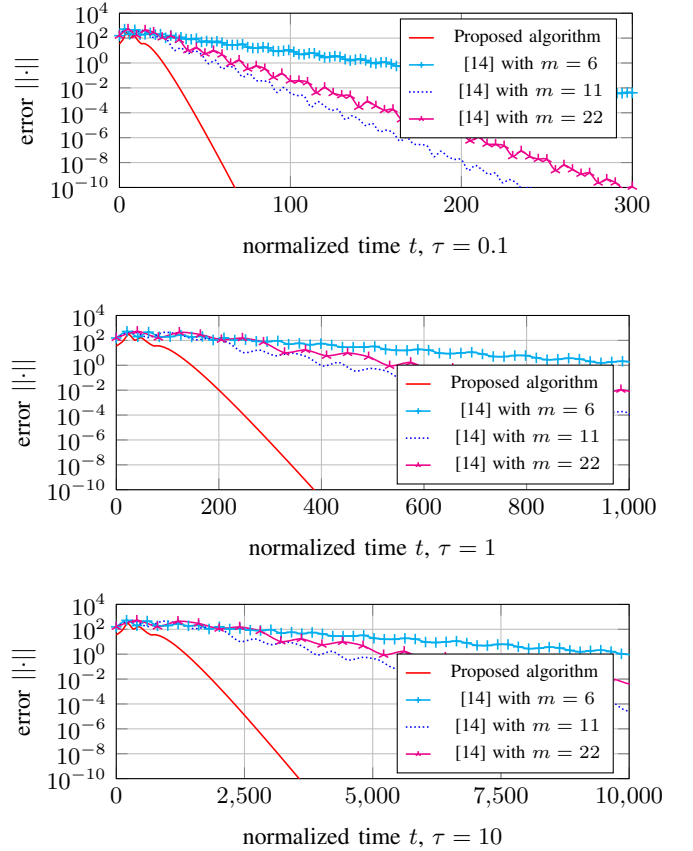

To conclude, in Fig.~\ref{fig:simulation_results} we provide a comparison\footnote{
Notice that also~\cite{wang2019distributed} relies on rounds of distributed averaging; however, such an approach only focuses on estimation, while no control action is considered. For this reason, and since in \cite{wang2019distributed} the estimation gains are designed in a way that would not be suitable for our approach (and vice versa), we opted not to compare them.} with the approach in Savas \emph{et al.}~\cite{savas2022separation}, which is based on $m$ repeated rounds of asymptotic consensus in-between state updates.
Following the approach in~\cite{scaman2017optimal}, by normalization of the time axis, we assume that the local computations executed at each estimation-control iteration $k$ are performed in one unit of time, denoted by $t$.
Moreover, we assume that each communication round within the asymptotic or time-varying consensus procedures requires a time equal to $\tau$ (which may be smaller, equal, or greater than 1). In particular, the upper plot of Fig.~\ref{fig:simulation_results} considers the case where $\tau=0.1$ (\ie communication is faster than local computation), the middle plot of Fig.~\ref{fig:simulation_results} considers the case where $\tau=1$ (\ie communication and computation are equally fast), while the case $\tau=10$ (\ie communication is slower than computation) is reported in the lower plot of Fig.~\ref{fig:simulation_results}.
In more detail, the figure considers the problem discussed in the previous example and compare the norm of the average error of our approach with the one in \cite{savas2022separation} relying, for fairness of comparison, on the same control and estimation gains chosen via our distributed initialization algorithm. The gains $W_{ij}$ were computed\footnote{In \cite{savas2022separation}, the gains $W_{ij}$ amount to the sum of two terms $W_{ij}^E=p_{ij}A$ and $W_{ij}^C$. In particular, the term $W_{ij}^C$ is computed by solving a centralized optimization problem featuring a positive definiteness constraint. To compute the terms $W_{ij}^E$, we resorted to an approximated solver, namely MIDACO~\cite{Schlueter_2013}. The interested reader is referred to~\cite{savas2022separation} for a complete discussion on how to compute the gains $W_{ij}$.} as in \cite{savas2022separation}.
Interestingly, while the proposed approach requires $\overline{m}=11$ communication rounds for each state update, the approach in \cite{savas2022separation}, for the considered instance, exhibits a converging behavior only for $m\geq 6$ asymptotic average consensus rounds. In this view, the figure reports the evolution of the error for the approach in \cite{savas2022separation} considering $m\in\{6,11,22\}$.
According to the figure, in all three considered scenarios for the computation/communication time, the proposed algorithm consistently outperforms the one in \cite{savas2022separation}. 
The comparison is particularly favorable to our algorithm when the communication time is comparable or larger than the local computation time, \ie $\tau \in \{1,10\}$.

Interestingly, the normalized time $t$ required for the case where $m=22$ is larger than the one required for $m=11$; in fact, we observe that for these two cases, the performance is comparable in terms of estimation-control iterations $k$, but setting $m=22$ requires twice as much communication rounds as for $m=11$. \textcolor{black}{Hence, in the case of asymptotic average-consensus rounds, increasing $m$ may not correspond to a performance increase. }

%%%%%%%%%%%%%%%%%%%%%%%%%%%%%%%%%%%%%%%%%%%
\section{Conclusions and Future Directions}\label{sec:conclusions}
\color{black}
This paper introduces a novel methodology for distributed LTI system estimation and control in discrete-time. It incorporates a finite-time average consensus mechanism between estimation steps, allowing for arbitrary eigenvalue placement in estimation error dynamics and overcoming the limitations of previous approaches. A novel, fully distributed method for computing control gains enhances the system's self-configuration and resilience. Simulation results demonstrate improved performance over other methods that are based on asymptotic consensus rounds.

Future research could explore: dealing with consensus mechanisms under unmet assumptions (e.g., when Assumption~\ref{assum:numberofsteps} is not satisfied), addressing communication challenges like delays, packet losses, and varying topologies, developing procedures for simultaneous gain design across nodes, and extending the framework to stochastic linear systems.
\color{black}

\bibliographystyle{IEEEtran}
\bibliography{references.bib}

\end{document}

%% file: finite-time-consensus-iterations.tex
\begin{tikzpicture}[>=stealth, thick,
   shorten >=3pt,
   shorten <=3pt,
%   node distance=2.5cm,
   on grid,
   auto,
  ]

  \tikzset{
    ultra thin/.style= {line width=0.1pt},
    very thin/.style=  {line width=0.2pt},
    thin/.style=       {line width=0.4pt},% thin is the default
    semithick/.style=  {line width=0.6pt},
    thick/.style=      {line width=0.8pt},
    very thick/.style= {line width=1.6pt},
    ultra thick/.style={line width=2pt}
}
  
\tikzset{node distance = 3cm and 2cm}
\tikzset{invincible/.style={draw=none, thin, circle, fill=none, minimum size=0}}
\tikzset{main/.style={draw=black, thick, draw, circle, fill=gray!20, minimum size=0.6cm}}
\tikzset{label/.style={draw=none, circle, fill=none, minimum size=0.8cm}}
\tikzset{edge/.style ={draw=black, thick, ->,> = latex'}}

%\tikzset{label/.style ={fill=white}}

\node[invincible] (0) at (0,0) {}; 

% main timeline
\path[-] +(-0.1,0) edge[very thick] +(1.5,0);
\path[-] +(1.4,0) edge[very thick, densely dashed] +(2.1,0);
\path[-] +(2,0) edge[very thick] +(4,0);
\path[-] +(3.9,0) edge[very thick, densely dashed] +(4.6,0);
\path[-] +(4.5,0) edge[very thick] +(6.5,0);
\path[arrows={-Latex[width=9pt, length=9pt]}] +(6.4,0) edge[very thick,densely dashed] +(7.3,0);

\path[-] +(0,0.35) edge[ultra thick,DarkRed!70] +(0,-0.35);
\path[-] +(0.5,0.25) edge[very thick,blue!70] +(0.5,-0.25);
\path[-] +(1,0.25) edge[very thick,blue!70] +(1,-0.25);

\path[-] +(2.5,0.25) edge[very thick,blue!70] +(2.5,-0.25);
\path[-] +(3,0.35) edge[ultra thick,DarkRed!70] +(3,-0.35);
\path[-] +(3.5,0.25) edge[very thick,blue!70] +(3.5,-0.25);

\path[-] +(5,0.25) edge[very thick,blue!70] +(5,-0.25);
\path[-] +(5.5,0.35) edge[ultra thick,DarkRed!70] +(5.5,-0.35);
\path[-] +(6,0.25) edge[very thick,blue!70] +(6,-0.25);

% dotted blue lines
\path[-] +(0,0.6) edge[densely dotted,thick,blue!50] +(0,0.2);
\path[-] +(1,0.6) edge[densely dotted,thick,blue!50] +(1,0.1);
\path[-] +(2.5,0.6) edge[densely dotted,thick,blue!50] +(2.5,0.1);
\path[-] +(5,0.6) edge[densely dotted,thick,blue!50] +(5,0.1);

% dotted red lines
\path[-] +(0,-0.6) edge[densely dotted,thick,DarkRed!50] +(0,-0.2);
\path[-] +(3,-0.6) edge[densely dotted,thick,DarkRed!50] +(3,-0.2);
\path[-] +(5.5,-0.6) edge[densely dotted,thick,DarkRed!50] +(5.5,-0.2);

% inner consensus loop labels
\node[label] at (-0.05,0.7) {\scriptsize $m=0$};
\node[label] at (0.95,0.7) {\scriptsize $m=2$};
\node[label] at (2.5,0.7) {\scriptsize $m=\bar{m}$};
\node[label] at (5,0.7) {\scriptsize $m=\bar{m}$};

% out 
\node[label] at (2.5,1) {\scriptsize $\overline{\bm{x}}[0]$};

\node[label] at (5,1) {\scriptsize $\overline{\bm{x}}[1]$};

\node[label] at (0,-0.7) {\scriptsize $\hat{\bm{x}}_l[k\!=\!0]$};
\node[label] at (3,-0.7) {\scriptsize $\hat{\bm{x}}_l[k\!=\!1]$};
\node[label] at (5.5,-0.7) {\scriptsize $\hat{\bm{x}}_l[k\!=\!2]$};

\node[label] at (7.75,0.05) {time};

\end{tikzpicture}

%% file: system_diagram.tex
\begin{tikzpicture}[auto, thick, node distance=1cm, >=triangle 45]

% Definition of blocks:
\tikzset{%
  block/.style    = {draw, rectangle,rounded corners=5, densely dotted, minimum height = 2.3em, minimum width = 3em},
  component/.style = {draw,rounded corners, thick, rectangle, minimum height = 1.9em, minimum width = 2em},
  sum/.style      = {draw, circle, node distance = 2cm}, % Adder
}
     % Drawing the blocks:
    \draw
    node at (0,-1.75) [component,rotate=90,minimum height=0.8cm,text width=1.3cm,align=center,fill=black!5] (system) {LTI System} % system
    node at (4,-0.5) [block,fill=blue!2,minimum width=5cm,minimum height=1.7cm] (v1) {}
    node at (4,-3) [block,fill=red!2,minimum width=5cm,minimum height=1.7cm] (vN) {}
    node at (4,-1.75) [block,fill=green!2,minimum width=1cm,minimum height=3.5cm] (consensus) {}

    % components blocks vN
    node at (2.2,-3) [component,fill=red!8] (CN) {$\texttt{C}_N$}
    node at (4,-3) [component,fill=red!8] (AN) {$\texttt{A}_N$}
    node at (5.7,-3) [component,fill=red!8] (EN) {$\texttt{E}_N$};
    % components connections
    \draw[->] (EN) -> node[xshift=5, yshift=0.2] {$\hat{x}_N$} (AN);
    \draw[->] (AN) -> node[xshift=0, yshift=0.2] {$\overline{x}$} (CN);
    \draw[-] (CN) -- node[below] {$u_N$} (0.1,-3);
    \draw[->] (0.1,-3) -- node[] {} (0.1,-2.52);
    % feedback
    \draw[-] (3.4,-3) -- node[] {} (3.4,-3.7);
    \draw[-] (3.4,-3.7) -- node[] {} (5.7,-3.7);
    \draw[<-] (EN) -- node[] {} (5.7,-3.7);

    % components blocks v1
    \draw
    node at (2.2,-0.5) [component,fill=blue!8] (C1) {$\texttt{C}_1$}
    node at (4,-0.5) [component,fill=blue!8] (A1) {$\texttt{A}_1$}
    node at (5.7,-0.5) [component,fill=blue!8] (E1) {$\texttt{E}_1$};
    % components connections
    \draw[->] (E1) -> node[xshift=5, yshift=-0.1] {$\hat{x}_1$} (A1);
    \draw[->] (A1) -> node[xshift=0, yshift=-0.1] {$\overline{x}$} (C1);
    \draw[-] (C1) -- node[above] {$u_1$} (0.1,-0.5);
    \draw[->] (0.1,-0.5) -- node[] {} (0.1,-0.98);
    % feedback
    \draw[-] (3.4,-0.5) -- node[] {} (3.4,0.2);
    \draw[-] (3.4,0.2) -- node[] {} (5.7,0.2);
    \draw[<-] (E1) -- node[] {} (5.7,0.2);

    \draw[-] (-0.1,0.5) -- node[] {} (7,0.5);
    \draw[-] (7,-0.5) -- node[] {} (7,0.5);
    \draw[->] (7,-0.5) -- node[xshift=8, yshift=0.2] {$y_1$} (E1);
	\draw[-] (-0.1,-0.99) -- node[] {} (-0.1,0.5);

    \draw[-] (-0.1,-4) -- node[] {} (7,-4);
    \draw[-] (7,-3) -- node[] {} (7,-4);
    \draw[->] (7,-3) -- node[xshift=8, yshift=13.5] {$y_N$} (EN);
	\draw[-] (-0.1,-2.51) -- node[] {} (-0.1,-4);

    \draw[->] (A1) edge[bend left=10,dashed] node[] {} (AN);
    \draw[->] (AN) edge[bend left=10,dashed] node[] {} (A1);

    \draw
    node[] at (2.25,-1.65) {$\vdots$};
    \draw
    node[] at (5.75,-1.65) {$\vdots$};
    
    \draw
    node[] at (2.25,0.15) {\footnotesize node $v_1$};
    \draw
    node[] at (2.25,-2.35) {\footnotesize node $v_N$};
\end{tikzpicture}

%% file: lastfigure01.tex
\definecolor{qual1}{HTML}{7FC97F} % green
\definecolor{qual2}{HTML}{BEAED4} % lila
\definecolor{qual3}{HTML}{FD673A} % orange
\definecolor{qual4}{HTML}{386CB0} % blue
\definecolor{qual5}{HTML}{E31A1C} % red
\newcommand\normX[1]{\left\lVert#1\right\rVert}
\begin{tikzpicture}
\begin{axis}[at={(0cm,0cm)},
            xmode=normal, ymode=log,
			at={(0cm,0cm)},
			grid=both,
			tick label style={font=\small},
    		label style={font=\small},
			width = 8.4cm, height=3.9cm,
			xlabel={normalized time $t$, $\tau=0.1$},
			ylabel={error $\lvert|\cdot\rvert|$},
			ymax=10000,
			xmax=300,
			ymin=0.0000000001,
			xtick={0,100,200,300},
			ytick={0.0000000001,0.00000001,0.000001, 0.0001, 0.01, 1, 100,10000},
			enlarge x limits=false,
            every axis plot/.append style={semithick},
            legend image post style={scale=0.55},
            legend pos=north east,
            smooth,
            legend style={font=\scriptsize, column sep=1ex,legend columns=1,mark repeat=1,mark phase=3},
            legend entries={Proposed algorithm, \cite{savas2022separation} with $m=6$, \cite{savas2022separation} with $m=11$, \cite{savas2022separation} with $m=22$}
 ]
\addplot [red,no markers] table [x expr=\thisrow{k1}*(11*0.1+1),y=nostro, col sep=comma]{xhatNormlastfigure.csv};
\addplot [cyan,mark=+] table [x expr=\thisrow{k1}*(20*0.1+1),y=l20, col sep=comma]{xhatNormlastfigure.csv};
\addplot [blue,densely dotted,no markers] table [x expr=\thisrow{k1}*(30*0.1+1),y=l30, col sep=comma]{xhatNormlastfigure.csv};
\addplot [magenta,mark=Mercedes star] table [x expr=\thisrow{k1}*(40*0.1+1),y=l40, col sep=comma]{xhatNormlastfigure.csv};
\end{axis}
\end{tikzpicture}

%% file: lastfigure1.tex
\definecolor{qual1}{HTML}{7FC97F} % green
\definecolor{qual2}{HTML}{BEAED4} % lila
\definecolor{qual3}{HTML}{FD673A} % orange
\definecolor{qual4}{HTML}{386CB0} % blue
\definecolor{qual5}{HTML}{E31A1C} % red
\newcommand\normX[1]{\left\lVert#1\right\rVert}
\begin{tikzpicture}
\begin{axis}[at={(0cm,0cm)},
            xmode=normal, ymode=log,
			at={(0cm,0cm)},
			grid=both,
			tick label style={font=\small},
    		label style={font=\small},
			width = 8.4cm, height=3.9cm,
			xlabel={normalized time $t$, $\tau=1$},
			ylabel={error $\lvert|\cdot\rvert|$},
			ymax=10000,
			xmax=1000,
			ymin=0.0000000001,
			ytick={0.0000000001,0.00000001,0.000001, 0.0001, 0.01, 1, 100,10000},
			enlarge x limits=false,
            every axis plot/.append style={semithick},
            legend image post style={scale=0.55},
            legend pos=south east,
            smooth,
            legend style={font=\scriptsize, column sep=1ex,legend columns=1,mark repeat=1,mark phase=3},
            legend entries={Proposed algorithm, \cite{savas2022separation} with $m=6$, \cite{savas2022separation} with $m=11$, \cite{savas2022separation} with $m=22$}
 ]
\addplot [red,no markers] table [x expr=\thisrow{k1}*(11*1+1),y=nostro, col sep=comma]{xhatNormlastfigure.csv};
\addplot [cyan,mark=+] table [x expr=\thisrow{k1}*(20*1+1),y=l20, col sep=comma]{xhatNormlastfigure.csv};
\addplot [blue,densely dotted,no markers] table [x expr=\thisrow{k1}*(30*1+1),y=l30, col sep=comma]{xhatNormlastfigure.csv};
\addplot [magenta,mark=Mercedes star] table [x expr=\thisrow{k1}*(40*1+1),y=l40, col sep=comma]{xhatNormlastfigure.csv};
\end{axis}
\end{tikzpicture}

%% file: lastfigure10.tex
\definecolor{qual1}{HTML}{7FC97F} % green
\definecolor{qual2}{HTML}{BEAED4} % lila
\definecolor{qual3}{HTML}{FD673A} % orange
\definecolor{qual4}{HTML}{386CB0} % blue
\definecolor{qual5}{HTML}{E31A1C} % red
\newcommand\normX[1]{\left\lVert#1\right\rVert}
\begin{tikzpicture}
\begin{axis}[at={(0cm,0cm)},
            xmode=normal, ymode=log,
			at={(0cm,0cm)},
			grid=both,
			tick label style={font=\small},
    		label style={font=\small},
			width = 8.4cm, height=3.9cm,
			xlabel={normalized time $t$, $\tau=10$},
			ylabel={error $\lvert|\cdot\rvert|$},
			ymax=10000,
			xmax=10000,
			ymin=0.0000000001,
			xtick={0,2500,5000,7500,10000},
			ytick={0.0000000001,0.00000001,0.000001, 0.0001, 0.01, 1, 100,10000},
            xticklabel style={
                    /pgf/number format/fixed,
                    /pgf/number format/precision=1
            },
            scaled x ticks=false,
			enlarge x limits=false,
            every axis plot/.append style={semithick},
            legend image post style={scale=0.55},
            legend pos=south east,
            smooth,
            legend style={font=\scriptsize, column sep=1ex,legend columns=1,mark repeat=1,mark phase=3},
            legend entries={Proposed algorithm, \cite{savas2022separation} with $m=6$, \cite{savas2022separation} with $m=11$, \cite{savas2022separation} with $m=22$}
 ]
\addplot [red,no markers] table [x expr=\thisrow{k1}*(11*10+1),y=nostro, col sep=comma]{xhatNormlastfigure.csv};
\addplot [cyan,mark=+] table [x expr=\thisrow{k1}*(20*10+1),y=l20, col sep=comma]{xhatNormlastfigure.csv};
\addplot [blue,densely dotted,no markers] table [x expr=\thisrow{k1}*(30*10+1),y=l30, col sep=comma]{xhatNormlastfigure.csv};
\addplot [magenta,mark=Mercedes star] table [x expr=\thisrow{k1}*(40*10+1),y=l40, col sep=comma]{xhatNormlastfigure.csv};
\end{axis}
\end{tikzpicture}

%% file: main.bbl
\begin{thebibliography}{10}
\providecommand{\url}[1]{#1}
\csname url@samestyle\endcsname
\providecommand{\newblock}{\relax}
\providecommand{\bibinfo}[2]{#2}
\providecommand{\BIBentrySTDinterwordspacing}{\spaceskip=0pt\relax}
\providecommand{\BIBentryALTinterwordstretchfactor}{4}
\providecommand{\BIBentryALTinterwordspacing}{\spaceskip=\fontdimen2\font plus
\BIBentryALTinterwordstretchfactor\fontdimen3\font minus
  \fontdimen4\font\relax}
\providecommand{\BIBforeignlanguage}[2]{{%
\expandafter\ifx\csname l@#1\endcsname\relax
\typeout{** WARNING: IEEEtran.bst: No hyphenation pattern has been}%
\typeout{** loaded for the language `#1'. Using the pattern for}%
\typeout{** the default language instead.}%
\else
\language=\csname l@#1\endcsname
\fi
#2}}
\providecommand{\BIBdecl}{\relax}
\BIBdecl

\bibitem{rego2019distributed}
F.~F. Rego, A.~M. Pascoal, A.~P. Aguiar, and C.~N. Jones, ``Distributed state
  estimation for discrete-time linear time invariant systems: A survey,''
  \emph{Annual Reviews in Control}, vol.~48, pp. 36--56, 2019.

\bibitem{Ugrinovskii:2013-conditions}
V.~Ugrinovskii, ``{Conditions for Detectability in Distributed Consensus-Based
  Observer Networks},'' \emph{IEEE Transactions on Automatic Control}, vol.~58,
  no.~10, pp. 2659--2664, 2013.

\bibitem{zhu2014cooperative}
H.~Zhu, K.~Liu, J.~L{\"u}, Z.~Lin, and Y.~Chen, ``On the cooperative
  observability of a continuous-time linear system on an undirected network,''
  in \emph{International Joint Conference on Neural Networks (IJCNN)}.\hskip
  1em plus 0.5em minus 0.4em\relax IEEE, 2014, pp. 2940--2944.

\bibitem{park2012necessary}
S.~Park and N.~C. Martins, ``{Necessary and sufficient conditions for the
  stabilizability of a class of LTI distributed observers},'' in \emph{IEEE
  Conference on Decision and Control (CDC)}.\hskip 1em plus 0.5em minus
  0.4em\relax IEEE, 2012, pp. 7431--7436.

\bibitem{park2016design}
------, ``{Design of distributed LTI observers for state omniscience},''
  \emph{IEEE Transactions on Automatic Control}, vol.~62, no.~2, pp. 561--576,
  2016.

\bibitem{mitra2018distributed}
A.~Mitra and S.~Sundaram, ``{Distributed observers for LTI systems},''
  \emph{IEEE Transactions on Automatic Control}, vol.~63, no.~11, pp.
  3689--3704, 2018.

\bibitem{khajenejad2023distributed}
M.~Khajenejad, S.~Brown, and S.~Mart{\'\i}nez, ``{Distributed interval
  observers for bounded-error LTI systems},'' in \emph{American Control
  Conference (ACC)}.\hskip 1em plus 0.5em minus 0.4em\relax IEEE, 2023, pp.
  2444--2449.

\bibitem{liu2023distributed}
T.~Liu and J.~Huang, ``{Distributed Exponential State Estimation for
  Discrete-time Linear Systems over Jointly Connected Switching Networks},''
  \emph{IEEE Transactions on Automatic Control}, 2023.

\bibitem{kim2016distributed}
T.~Kim, H.~Shim, and D.~D. Cho, ``Distributed luenberger observer design,'' in
  \emph{IEEE Conference on Decision and Control (CDC)}.\hskip 1em plus 0.5em
  minus 0.4em\relax IEEE, 2016, pp. 6928--6933.

\bibitem{han2018simple}
W.~Han, H.~L. Trentelman, Z.~Wang, and Y.~Shen, ``A simple approach to
  distributed observer design for linear systems,'' \emph{IEEE Transactions on
  Automatic Control}, vol.~64, no.~1, pp. 329--336, 2018.

\bibitem{wang2019distributed}
L.~Wang, J.~Liu, and A.~S. Morse, ``A distributed observer for a
  continuous-time linear system,'' in \emph{American Control Conference
  (ACC)}.\hskip 1em plus 0.5em minus 0.4em\relax IEEE, 2019, pp. 86--91.

\bibitem{wang2019distributedb}
L.~Wang, J.~Liu, A.~S. Morse, and B.~D. Anderson, ``A distributed observer for
  a discrete-time linear system,'' in \emph{IEEE Conference on Decision and
  Control (CDC)}, 2019, pp. 367--372.

\bibitem{wang2020distributed}
L.~Wang, D.~Fullmer, F.~Liu, and A.~S. Morse, ``Distributed control of linear
  multi-channel systems: summary of results,'' in \emph{American Control
  Conference (ACC)}.\hskip 1em plus 0.5em minus 0.4em\relax IEEE, 2020, pp.
  4576--4581.

\bibitem{savas2022separation}
A.~J. Savas, S.~Park, H.~V. Poor, and N.~E. Leonard, ``{On Separation of
  Distributed Estimation and Control for LTI Systems},'' in \emph{IEEE
  Conference on Decision and Control (CDC)}.\hskip 1em plus 0.5em minus
  0.4em\relax IEEE, 2022, pp. 963--968.

\bibitem{kim2019completely}
T.~Kim, C.~Lee, and H.~Shim, ``Completely decentralized design of distributed
  observer for linear systems,'' \emph{IEEE Transactions on Automatic Control},
  vol.~65, no.~11, pp. 4664--4678, 2019.

\bibitem{2010Christoforos:RC}
A.~D. Dominguez-Garcia and C.~N. Hadjicostis, ``Coordination and control of
  distributed energy resources for provision of ancillary services,'' in
  \emph{IEEE International Conference on Smart Grid Communications}, 2010, pp.
  537--542.

\bibitem{charalambous2015distributed}
T.~Charalambous, Y.~Yuan, T.~Yang, W.~Pan, C.~N. Hadjicostis, and M.~Johansson,
  ``Distributed finite-time average consensus in digraphs in the presence of
  time delays,'' \emph{IEEE Transactions on Control of Network Systems},
  vol.~2, no.~4, pp. 370--381, 2015.

\bibitem{2013:Ye}
Y.~Yuan, G.-B. Stan, M.~Barahona, L.~Shi, and J.~Gon\c{c}alves, ``Decentralised
  minimal-time consensus,'' \emph{Automatica}, vol.~49, no.~5, pp. 1227--1235,
  2013.

\bibitem{themis:2018ECC_termination}
T.~Charalambous and C.~N. Hadjicostis, ``When to stop iterating in digraphs of
  unknown size? an application to finite-time average consensus,'' in
  \emph{European Control Conference (ECC)}, 2018, pp. 1--7.

\bibitem{2008:Cortes}
J.~Cort\'{e}s, ``Distributed algorithms for reaching consensus on general
  functions,'' \emph{Automatica}, vol.~44, pp. 726--737, March 2008.

\bibitem{nejad2009max}
B.~M. Nejad, S.~A. Attia, and J.~Raisch, ``Max-consensus in a max-plus
  algebraic setting: The case of fixed communication topologies,'' in
  \emph{International Symposium on Information, Communication and Automation
  Technologies}.\hskip 1em plus 0.5em minus 0.4em\relax IEEE, 2009, pp. 1--7.

\bibitem{luenberger1979introduction}
D.~G. Luenberger, \emph{Introduction to dynamic systems: theory, models, and
  applications}.\hskip 1em plus 0.5em minus 0.4em\relax Wiley New York, 1979,
  vol.~1.

\bibitem{mackey2008deflation}
L.~Mackey, ``Deflation methods for sparse pca,'' \emph{Advances in neural
  information processing systems}, vol.~21, 2008.

\bibitem{olfati2004consensus}
R.~Olfati-Saber and R.~M. Murray, ``Consensus problems in networks of agents
  with switching topology and time-delays,'' \emph{IEEE Transactions on
  Automatic Control}, vol.~49, no.~9, pp. 1520--1533, 2004.

\bibitem{muniraju2019analysis}
G.~Muniraju, C.~Tepedelenlioglu, and A.~Spanias, ``Analysis and design of
  robust max consensus for wireless sensor networks,'' \emph{IEEE Transactions
  on Signal and Information Processing over Networks}, vol.~5, no.~4, pp.
  779--791, 2019.

\bibitem{scaman2017optimal}
K.~Scaman, F.~Bach, S.~Bubeck, Y.~T. Lee, and L.~Massouli{\'e}, ``Optimal
  algorithms for smooth and strongly convex distributed optimization in
  networks,'' in \emph{international conference on machine learning}.\hskip 1em
  plus 0.5em minus 0.4em\relax PMLR, 2017, pp. 3027--3036.

\bibitem{Schlueter_2013}
M.~Schlueter, S.~O. Erb, M.~Gerdts, S.~Kemble, and J.-J. R{\"u}ckmann, ``Midaco
  on minlp space applications,'' \emph{Advances in Space Research}, vol.~51,
  no.~7, pp. 1116--1131, 2013.

\end{thebibliography}
